\documentclass[conference]{IEEEtran}
\IEEEoverridecommandlockouts

\usepackage{cite}

\usepackage{amsmath,amssymb,amsfonts}
\usepackage{mathtools}
\usepackage{amsthm}
\newtheorem{lemma}{Lemma}

\usepackage{graphicx}
\usepackage{xcolor}
\usepackage{booktabs}
\usepackage{longtable}
\usepackage{array}
\usepackage{makecell}
\usepackage{pdflscape}
\usepackage{pgfplots}
\pgfplotsset{compat=1.18}

\usepackage{algorithm}
\usepackage{algpseudocode}

\usepackage{listings}
\usepackage{pifont}
\usepackage{textcomp}
\usepackage{multicol}
\usepackage{ifthen}
\usepackage{soul}
\usepackage{comment}

\usepackage{balance}
\usepackage[colorlinks=true, citecolor=blue]{hyperref}
\def\BibTeX{{\rm B\kern-.05em{\sc i\kern-.025em b}\kern-.08em
    T\kern-.1667em\lower.7ex\hbox{E}\kern-.125emX}}
\begin{document}

\title{Neural Stringology Based Cryptanalysis of EChaCha20
\thanks{This paper has been accepted for publication in "The 2nd International Conference on Sustainability, Innovation, and Society (ICSIS 2026)", Valencia, Spain}
}
\author{ 
\IEEEauthorblockN{Victor Kebande\IEEEauthorrefmark{1}\IEEEauthorrefmark{2}}
\IEEEauthorblockA{\IEEEauthorrefmark{1}University of Colorado Denver, CO USA\\
\IEEEauthorrefmark{2}ATLAS Institute, University of Colorado 
Boulder, Colorado, USA\\
Email: victor.kebande@ucdenver.edu, victor.kebande@colorado.edu}
} 

\maketitle

\begin{abstract}

Modern stream ciphers rely on strong diffusion and Pseudorandom Keystream Generation (PKG) to resist cryptanalysis. While conventional evaluation methods such as statistical randomness tests and differential analysis provide important security assurances, they may fail to detect localized structural patterns embedded within cipher outputs. In this paper, a Neural Stringology Cryptanalysis (NSC) framework that combines classical string pattern analysis with Machine Learning techniques to investigate potential structural anomalies in stream cipher keystreams is introduced. The proposed approach first applies stringology-inspired feature extraction methods like $m$-gram frequency analysis, substring recurrence detection, and positional pattern statistics aligned with the internal operations of Add Rotate XOR (ARX)-based stream ciphers. These extracted features are then analyzed using a neural learning model to identify deviations from expected random behavior and to detect subtle structural patterns that may not be captured by traditional statistical tests. Experimental evaluation is conducted on keystream outputs generated by the \textit{EChaCha20} stream cipher under multiple configurations, including reduced-round variants. The results demonstrate that the proposed NSC framework can identify distinguishable structural characteristics in the keystream data under controlled conditions, suggesting that integrating machine learning with stringology-based analysis provides a promising complementary methodology for evaluating the structural robustness of modern ARX-based stream cipher designs.

\end{abstract}


\section{Introduction}

Modern stream ciphers play a critical role in securing contemporary digital communication systems, including network protocols, infrastructures, and secure messaging platforms \cite{kuznetsov2024high}. Cipher designs like the ChaCha family  have gained widespread adoption because of their strong diffusion properties, their resistance to timing attacks and high performance across diverse hardware environments according to  \cite{bernstein2008chacha}. In particular, due to the Add-Rotate-XOR (ARX) based constructions, \cite{song2020secure}, they have become attractive due to their simplicity and efficiency while maintaining strong cryptographic security guarantees. For example, the  improved constructions like the \textit{EChaCha20} stream cipher extends the ChaCha design by modifying its Internal State $(S_t)$ structure and transformation mechanisms in order to improve diffusion and strengthen resistance against structural cryptanalysis \cite{kebande2023extended}.

The process of evaluating the security of stream ciphers typically relies on a combination of statistical randomness tests, differential cryptanalysis, and algebraic analysis. Also, standard evaluation suites such as the NIST Statistical Test Suite (STS), TestU01 \cite{sleem2020testu01}, PractRand are also commonly used to assess the statistical quality of generated keystreams \cite{sleem2020testu01, zubkov2019testing}. It has been seen from literature that these approaches provide important insights into global randomness properties. However, these standards   are not always designed to detect localized structural irregularities that may arise from internal cipher operations. 


In the recent past, stringology-based approaches have emerged as a promising complementary technique for analyzing sequential data. Stringology is the study of efficient pattern matching and substring analysis algorithms \cite{watson2012correctness}. It has traditionally been applied in domains such as text processing, bioinformatics, and data mining. Algorithms such as Knuth–Morris–Pratt (KMP) \cite{shapira2006adapting} and Boyer–Moore (BM) \cite{hyyro2004boyer} enable efficient detection of recurring patterns and structural repetitions in large datasets \cite{watson2003boyer}. When applied to cryptographic outputs, these techniques provide a novel perspective by treating keystream sequences as structured strings and searching for repeated patterns, substring correlations, and structural artifacts that may indicate deviations from ideal randomness.

Despite the effectiveness of classical stringology methods for pattern detection, identifying complex structural relationships in large-scale keystream datasets can remain challenging. Machine Learning (ML) techniques have recently demonstrated strong capabilities in discovering hidden patterns in high-dimensional data \cite{verleysen2003learning}. Integrating ML with string-based pattern analysis therefore offers a new opportunity to enhance cryptanalytic evaluation methods. 


In this paper, we introduce a \textit{Neural Stringology Cryptanalysis} (NSC) framework that combines classical string pattern analysis with machine learning techniques to investigate structural properties of stream cipher keystreams. The proposed approach first extracts stringology-inspired features from keystream outputs, including $m$-gram frequency distributions, substring recurrence patterns, and positional pattern statistics aligned with the internal operations of ARX-based stream ciphers. These features are subsequently analyzed using a neural learning model to identify deviations from expected random behavior and to evaluate the structural robustness of cipher outputs.

To evaluate the effectiveness of the NSC framework, experiments are conducted on keystream datasets generated by the \textit{EChaCha20} stream cipher under multiple configurations, including reduced-round variants. The experimental results have demonstrated that the integration of ML with stringology-based feature extraction can assist in identifying distinguishable structural characteristics in keystream outputs under controlled experimental conditions. These findings suggest that neural stringology analysis can serve as a complementary tool for evaluating the structural robustness of modern ARX-based stream cipher constructions alongside traditional statistical and differential cryptanalysis techniques.

The remainder of this paper is organized as follows. Section II  presents background on stream ciphers, ARX constructions, and stringology-based analysis methods. Section III introduces the related work while Section IV and V formulates the problem and presents the threat model respectively. This is followed by the NSC framework in Section VI. Experiments, results and analysis are discussed in Section VII and VIII respectively. After this, a discussion is given in Section IX followed by a conclusion and a mention of future work in Section X.


\section{Background}

\subsection{Stream Ciphers Construction }

Stream ciphers are symmetric encryption algorithms that generate a pseudorandom keystream which is combined with plaintext using simple operations such as XOR to produce ciphertext \cite{jiao2020stream}. Their efficiency and low computational overhead make them suitable for high-speed communication systems, including network protocols and Internet-of-Things (IoT) devices. A secure stream cipher must produce a keystream that is computationally indistinguishable from a truly random sequence, ensuring confidentiality even when large amounts of data are encrypted as is shown in Figure \ref{fig:exampless}.

\begin{figure}[h]
\centering
\includegraphics[width=0.8\linewidth]{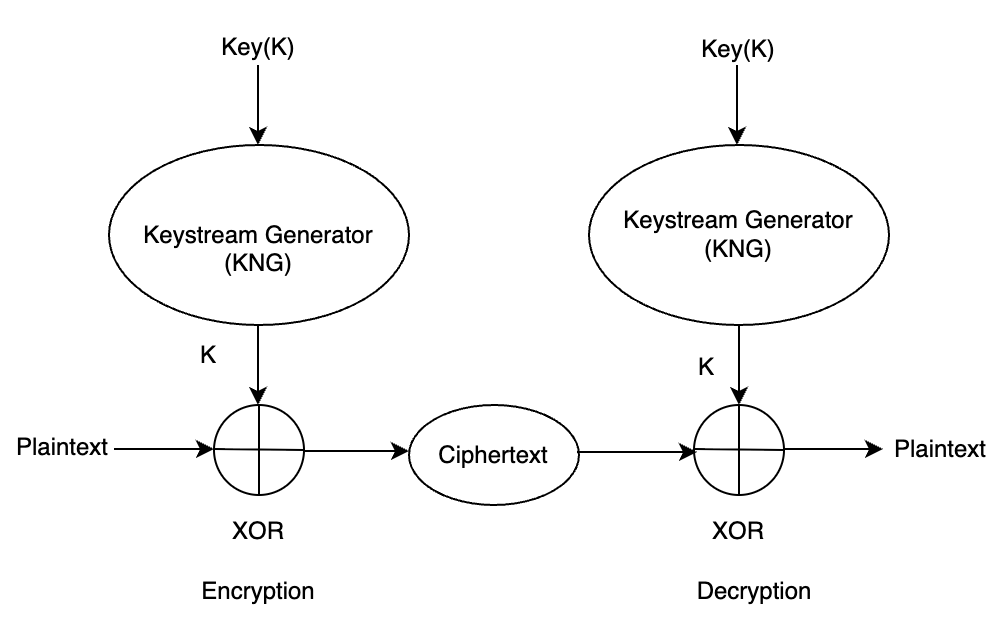}
\caption{Steam Cipher}
\label{fig:exampless}
\end{figure}

Many modern stream ciphers employ \textit{Add–Rotate–XOR (ARX)} constructions, which rely on the combination of modular addition, bitwise rotation, and XOR operations \cite{perrin2024addition}. ARX designs are attractive because they avoid complex substitution tables and can be efficiently implemented on general-purpose processors. Well-known examples include the ChaCha family of stream ciphers \cite{bernstein2008chacha}, where diffusion and nonlinearity are achieved through repeated ARX transformations applied to internal state matrices. Due to their structured word-level operations, ARX-based ciphers are often analyzed for potential rotational patterns or structural artifacts in their keystream outputs.

\subsection{ EchaCha20}
\textit{EChaCha20} is a stream cipher derived from the ChaCha family of ARX-based constructions \cite{kebande2023extended}. Similar to ChaCha20 \cite{bernstein2008chacha}, the cipher generates a pseudorandom keystream which is combined with plaintext using XOR operations to produce ciphertext. The internal design relies on Add–Rotate–XOR (ARX) transformations applied to a structured internal state composed of multiple word-level registers. These operations are repeated through several rounds in order to achieve strong diffusion and nonlinearity across the internal state. It enhances the quarter-round function (QR-F) and incorporates 32-bit input words along with ARX operations. The internal state of \textit{EChaCha20} is initialized using a secret key $K$, a nonce $N$, and a block counter.

The fundamental transformation used in ChaCha-based constructions,
including \textit{EChaCha20}, is the \textit{quarter-round} function. This
operation updates four 32-bit state words using a sequence of
Add–Rotate–XOR (ARX) operations. Let the four state words be
denoted by $(a,b,c,d)$ \cite{kebande2023extended}. The quarter-round transformation is defined
as follows:

Taking \textit{EChaCha20} input for $x=(x_0,x_1,x_2,x_3,x_4,x_5)$, then the
$QR-Fs$ will be given by $y=(y_0,y_1,y_2,y_3,y_4,y_5)$, where $x$
and $y$ are six-word input and output, respectively. Therefore, the
extended $QR-Fs$ is given as shown in the word computation:

$z_1 = y_1 \oplus ((y_0 + y_5)\ll 2)$,
$z_2 = y_2 \oplus ((z_1 + y_3)\ll 4)$,
$z_3 = y_3 \oplus ((z_2 + y_4)\ll 7)$,
$z_4 = y_4 \oplus ((z_3 + y_0)\ll 8)$,
$z_5 = y_5 \oplus ((z_4 + z_1)\ll 12)$,
$z_0 = y_0 \oplus ((z_5 + z_2)\ll 16)$.

where $+$ denotes addition modulo $2^{32}$, $\oplus$ denotes the
bitwise XOR operation, and $\lll r$ represents a left rotation by
$r$ bits \cite{kebande2023extended}.


Because \textit{EChaCha20} relies on deterministic word-level transformations, the generated keystream can be viewed as a sequential symbolic string whose structural properties may reflect characteristics of the underlying cipher operations. This makes \textit{EChaCha20} a suitable candidate for evaluating the proposed neural stringology cryptanalysis framework, which analyzes keystream outputs using pattern-based feature extraction and machine learning techniques in order to detect subtle structural signals.

\subsection{Stringology and Pattern Matching}

Stringology is the study of algorithms and data structures designed for efficient processing and analysis of strings \cite{watson2012correctness}. Classical string matching techniques have been widely used in applications such as text processing, bioinformatics, and data mining, where identifying repeated patterns and structural similarities within sequences is essential. When applied to cryptographic analysis, keystream outputs can be interpreted as long sequences of symbols, allowing pattern detection techniques to reveal potential irregularities or recurring structures.

Two well-known pattern matching algorithms are the KMP and BM algorithms \cite{shapira2006adapting, watson2003boyer}. The KMP algorithm performs efficient substring searches by precomputing prefix information that avoids redundant comparisons, while the BM algorithm accelerates searches by using heuristic skipping rules based on character mismatches. These algorithms enable efficient detection of repeated substrings and structural correlations in large datasets. In the context of cryptographic evaluation, such techniques can be used to examine keystream sequences for pattern frequencies, substring recurrences, and other indicators of non-random behavior.

\subsection{Machine Learning in Cryptanalysis}

Machine Learning (ML) techniques have recently attracted attention as tools for analyzing complex patterns in security-related datasets \cite{liu2018survey}. In cryptanalysis, ML models can be trained to detect subtle statistical relationships that may not be easily observable using traditional analytical methods \cite{verleysen2003learning}. Applications  include neural distinguishers, side-channel analysis, and anomaly detection in cryptographic outputs.

A common approach involves extracting descriptive features from cryptographic data and using ML algorithms to identify deviations from expected random behavior \cite{barbosa2016machine}. By learning relationships among these features, the models may assist in identifying structural patterns or biases present in cipher outputs. Integrating ML with traditional analytical techniques therefore offers a complementary perspective for evaluating the robustness of modern cryptographic constructions.

\section{Related Work}

Research on stream cipher evaluation has traditionally relied on statistical randomness testing, differential cryptanalysis, and algebraic analysis. Standard test suites such as NIST SP 800-22, TestU01, and PractRand are widely used to assess the pseudorandomness of keystream outputs \cite{sleem2020testu01, sys2014faster}. These approaches are valuable for detecting global statistical weaknesses; however, they are generally less suited for identifying localized structural patterns that may arise from word-level operations in ARX-based ciphers \cite{barbero2024overview}. This limitation is particularly relevant when analyzing modern designs whose security depends on subtle diffusion and rotation behavior.

Within the family of ARX stream ciphers \cite{barbero2024overview}, ChaCha20 \cite{bernstein2008chacha} and its variants have received significant attention due to their efficiency and strong practical security properties. Prior studies on ChaCha-style constructions have investigated diffusion, avalanche behavior, and resistance to rotational cryptanalysis, with particular focus on how rotation constants and round structure influence differential propagation \cite{bernstein2008chacha,barbero2022rotational}. In the case of EChaCha, existing evaluation has largely centered on statistical randomness tests and structural analysis of its expanded state and modified quarter-round function, showing improved diffusion and strong resistance within the tested bounds \cite{kebande2023extended}. Nevertheless, these studies do not explicitly integrate data-driven learning mechanisms into the analysis process. 

Pattern matching and string processing methods have long been used in text algorithms \cite{watson2012correctness}, bioinformatics, and secure data processing. Classical algorithms such as Knuth--Morris--Pratt and Boyer--Moore provide efficient mechanisms for detecting repeated substrings, structural recurrences, and alignment patterns in large sequential datasets \cite{kourie2011compile, shapira2006adapting, hyyro2004boyer}. Recent stringology-inspired studies suggest that such methods can be adapted to security-oriented contexts, but their use in cryptanalysis remains relatively limited \cite{kebande2026stringologybasedcryptanalysisechacha20stream}. In particular, prior work has not fully explored how stringology-based feature extraction can be combined with modern learning techniques for stream cipher evaluation. 

In parallel, machine learning has increasingly been applied in cryptographic research, particularly in areas such as side-channel analysis, neural distinguishers, and anomaly detection. Learning-based methods have shown promise in identifying subtle patterns that may be difficult to capture using handcrafted analytical rules alone. Unsupervised approaches, including clustering and autoencoders, have also been suggested as useful tools for discovering rare or non-obvious structures in cryptographic data. However, much of this work focuses either on implementation leakage or generic statistical classification, rather than on stringology-guided structural analysis of ARX keystreams. 

This work differs from prior work in two ways. First, it builds on stringology-based cryptanalysis by treating keystream evaluation as a structured pattern-search problem aligned with the internal word-level behavior of ARX stream ciphers. Second, it extends that framework by incorporating machine learning as an analysis layer over stringology-derived features. In this way, the proposed approach connects exact pattern matching with learning-based structural detection, providing a complementary methodology for evaluating \textit{EChaCha20} and related stream cipher constructions.

\section{Problem Formulation}

In the modern stream cipher evaluation, the primary security objective is that the keystream generated by the cipher should be computationally indistinguishable from a truly random sequence \cite{berbain2006quad}. Let $K \in \{0,1\}^{\kappa}$ denote the secret key and $N \in \{0,1\}^{\nu}$ denote the nonce used to initialize the cipher. The keystream generator of the cipher can be modeled as a deterministic function as is shown in Eq 1.

\begin{equation}
\mathcal{G}(K,N) \rightarrow S
\label{eq:keystreamgen}
\end{equation}

where $S = s_1 s_2 \ldots s_n \in \{0,1\}^n$ represents the output keystream of length $n$. In a secure stream cipher, the distribution of $S$ should be indistinguishable from the uniform distribution $U_n$ over $\{0,1\}^n$.

Traditional cryptanalytic evaluation methods rely on statistical randomness tests or differential analysis to assess whether the output sequence deviates from ideal random behavior. However, these techniques often evaluate aggregate statistical properties and may fail to capture localized structural characteristics that arise from internal cipher transformations. In ARX-based constructions such as \textit{EChaCha20} \cite{kebande2023extended}, the keystream is generated through repeated applications of modular addition, bit rotations, and XOR operations over word-level state matrices. These transformations may introduce subtle structural correlations within the generated sequence that are difficult to detect using conventional statistical tests.

To analyze such structural characteristics, this work considers the keystream sequence $S$ as a symbolic string that can be analyzed using stringology-based pattern detection techniques. Let $P \in \{0,1\}^{m}$ denote a substring pattern of length $m$. The number of occurrences of pattern $P$ in the keystream $S$ can be expressed as is shown in Eq 2.

\begin{multline}
Adv_{\mathcal{A}} =
\left|
Pr[\mathcal{A}(S)=1 \mid S \leftarrow \mathcal{G}(K,N)]
\right. \\
\left.
- Pr[\mathcal{A}(S)=1 \mid S \leftarrow U_n]
\right|
\end{multline}

For an ideal random sequence, the expected probability of observing pattern $P$ is approximately $2^{-m}$. Significant deviations from this probability may indicate the presence of structural patterns or correlations in the keystream output.

The central problem addressed in this work is to determine whether such structural patterns can be systematically detected through a combination of stringology-based feature extraction and machine learning analysis. In particular, we investigate whether features derived from substring frequency distributions, recurrence patterns, and positional correlations can reveal measurable differences between cipher-generated keystream sequences and truly random sequences.

Formally, consider an adversary $\mathcal{A}$ that is given access to an oracle $\mathcal{O}$ which outputs a sequence $S$ drawn either from the keystream generator $\mathcal{G}(K,N)$ or from the uniform distribution $U_n$. The goal of the adversary is to determine the origin of the observed sequence. The distinguishing advantage of the adversary is defined as is shown in Eq 3.

\begin{equation}
Adv_{\mathcal{A}} =
\left|
\Pr[\mathcal{A}(S_c)=1] -
\Pr[\mathcal{A}(S_r)=1]
\right|
\end{equation}

If the adversary can achieve a non-negligible advantage in distinguishing the two distributions using features derived from stringology-based analysis, then the keystream may exhibit detectable structural properties that deviate from ideal randomness.

In the suggest approach, we  investigate whether the integration of machine learning with stringology-based pattern analysis can construct a practical distinguisher for \textit{EChaCha20} keystream outputs under controlled experimental conditions. By formulating keystream analysis as a pattern detection and learning problem, the proposed approach seeks to provide a complementary perspective for evaluating the structural robustness of modern ARX stream cipher constructions.

\section{Threat Model}

Our threat model considers a consider a standard cryptanalytic threat model in which an adversary attempts to distinguish cipher-generated keystream sequences from truly random sequences. The objective of the adversary is not to recover the secret key, but rather to determine whether observable structural properties of the keystream reveal deviations from ideal randomness.

Let $K \in \{0,1\}^{\kappa}$ denote the secret key and $N \in \{0,1\}^{\nu}$ denote the nonce used to initialize the \textit{EChaCha20} stream cipher. The keystream generator is modeled as a deterministic function as shown in Eq 4.  

\begin{equation}
S \leftarrow \mathcal{G}(K,N)
\end{equation}

where $S \in \{0,1\}^{n}$ represents the generated keystream sequence. In a secure stream cipher, the distribution of $S$ should be computationally indistinguishable from the uniform distribution $U_n$ over $\{0,1\}^n$.

The adversary $\mathcal{A}$ is assumed to have oracle access to a sequence generator $\mathcal{O}$ that outputs either a cipher-generated sequence $S \leftarrow G(K,N)$ or a uniformly random sequence $S \leftarrow U_n$. The adversary does not have access to the secret key $K$, the nonce $N$, or the internal state of the cipher. Instead, the adversary observes only the output sequence $S$ and attempts to determine its origin.

Within the proposed framework, the adversary constructs a feature representation $\Phi(S)$ using stringology-based pattern extraction techniques such as $m$-gram frequency analysis, substring recurrence detection, and positional pattern statistics. These features are then analyzed using a neural learning model $M$ which outputs a binary classification indicating whether the sequence is likely to originate from the cipher generator or from a uniform random source.

The distinguishing advantage of the adversary is defined as

\[
Adv_{\mathcal{A}} =
\left|
Pr[\mathcal{A}(S)=1 \mid S \leftarrow G(K,N)]
-
Pr[\mathcal{A}(S)=1 \mid S \leftarrow U_n]
\right|.
\]

If the adversary achieves a non-negligible advantage, the keystream sequence exhibits detectable structural characteristics under the proposed analysis framework.

It is important to emphasize that this threat model focuses exclusively on statistical distinguishability and structural analysis of keystream outputs \cite{lv2013distinguishing}. The adversary does not perform key-recovery attacks, side-channel analysis, or implementation-based attacks. Instead, the goal is to evaluate whether structural signals produced by the internal transformations of the cipher can be detected through stringology-based pattern analysis combined with machine learning techniques.

\section{Neural Stringology Cryptanalysis Framework for \textit{EChaCha20}}

This section presents the proposed \textit{Neural Stringology Cryptanalysis} (NSC) framework for analyzing structural properties of the \textit{EChaCha20} stream cipher. The NSC framework combines classical stringology-based pattern analysis with ML methods to investigate whether structural signals embedded in keystream outputs can be systematically detected. Unlike traditional statistical test suites that primarily examine global randomness properties, the NSC framework focuses on localized structural characteristics that arise from the internal Add--Rotate--XOR (ARX) transformations of the cipher. By interpreting the keystream as a structured sequence, the framework allows pattern detection algorithms and learning models to jointly evaluate the structural behavior of \textit{EChaCha20} outputs under controlled experimental conditions.

\subsection{Architecture of the NSC Framework}

Let $K \in \{0,1\}^{256}$ denote the secret key and $N \in \{0,1\}^{128}$ the nonce used to initialize the \textit{EChaCha20} cipher. The keystream generator can be modeled as a deterministic function

\[
\mathcal{G}(K,N) \rightarrow S
\]

where $S = s_1 s_2 \ldots s_n$ represents the generated keystream sequence of length $n$. In the proposed framework, the sequence $S$ is treated as a symbolic string whose structural properties can be analyzed using stringology-based algorithms. The analysis pipeline maps the keystream sequence to a feature representation that captures pattern frequencies, substring recurrences, and structural correlations.

Formally, the NSC framework defines a feature extraction mapping

\[
\Phi : \{0,1\}^{n} \rightarrow \mathbb{R}^{d}
\]

which transforms the keystream sequence $S$ into a feature vector $X = \Phi(S)$ of dimension $d$. These features encode structural characteristics of the keystream that may reflect diffusion properties or internal transformation patterns of the cipher. The resulting feature vectors are subsequently analyzed using machine learning models that attempt to identify deviations from ideal random behavior.

\subsection{Stringology Feature Extraction}

In this study, the keystream sequence is analyzed using classical stringology techniques that detect repeated substrings and structural recurrences. Let $P$ denote a substring pattern of length $m$ extracted from the keystream sequence $S$. The occurrence frequency of pattern $P$ is defined as is shown in Eq 5.

\begin{equation}
f(P,S) = \left| \{\, i \mid S_{i:i+m-1} = P \,\} \right|
\label{eq:patterncount}
\end{equation}

where $S_{i..i+m-1}$ denotes the substring beginning at position $i$. For an ideal random sequence, the expected probability of observing a specific pattern $P$ of length $m$ is approximately $2^{-m}$. Deviations from this expected frequency may indicate localized structural behavior in the keystream.

To capture such properties, the framework evaluates $m$-gram distributions for $m \in \{8,16,32\}$, along with substring recurrence statistics and positional pattern densities across the keystream sequence. These statistics serve as descriptive features that characterize the structural properties of the generated sequence and provide inputs for subsequent learning-based analysis.

\subsection{Pattern Representation}

The extracted pattern statistics are transformed into structured feature representations suitable for machine learning analysis. For a given keystream segment $S$, the feature representation is defined as is shown in Eq 6

\begin{equation}
X = (x_1, x_2, \ldots, x_d)
\end{equation}

where each component $x_i$ represents a normalized statistic derived from the observed pattern distributions, recurrence rates, or positional frequency measurements. This representation converts raw keystream sequences into numerical vectors that encode structural characteristics of the cipher output.

From a cryptanalytic perspective, these feature vectors provide a compact representation of the statistical behavior of the keystream. By comparing feature vectors derived from different keystream sources, it becomes possible to analyze whether the structural properties of \textit{EChaCha20} outputs differ from those of truly random sequences.

\subsection{AI Integration}

Machine learning models are integrated into the framework as an analytical layer that evaluates the extracted feature representations. Let $\mathcal{M}$ denote a learning model trained on a dataset of feature vectors derived from keystream sequences. The model attempts to learn a classification function

\[
\mathcal{M} : \mathbb{R}^{d} \rightarrow \{0,1\}
\]

that distinguishes between feature vectors derived from cipher-generated sequences and those derived from random sources.

From a cryptanalytic perspective, the learning model can be interpreted as an adversarial distinguisher. Consider an adversary $\mathcal{A}$ that is given access to an oracle $\mathcal{O}$ which outputs either a keystream generated by \textit{EChaCha20} or a truly random sequence of equal length. The goal of the adversary is to determine which distribution produced the observed sequence. The distinguishing advantage of the adversary is defined as

\[
Adv_{\mathcal{A}} = \left| \Pr[\mathcal{A}(S) = 1 \mid S \leftarrow \mathcal{G}(K,N)] -
\Pr[\mathcal{A}(S) = 1 \mid S \leftarrow U_n] \right|
\]

where $U_n$ denotes the uniform distribution over $\{0,1\}^{n}$.

\begin{lemma}
If the machine learning model $\mathcal{M}$ achieves classification accuracy significantly greater than random guessing on feature vectors derived from $\Phi(S)$, then the induced adversary $\mathcal{A}$ forms a statistical distinguisher between \textit{EChaCha20} keystream outputs and uniform random sequences.
\end{lemma}

\textit{Proof Sketch.} If the feature extraction function $\Phi$ preserves structural information about the keystream and the learning model $\mathcal{M}$ can correctly classify sequences with probability greater than $1/2 + \epsilon$ for some non-negligible $\epsilon$, then the adversary constructed from $\mathcal{M}$ has advantage $\epsilon$ in distinguishing the two distributions. This implies that the feature representation captures measurable structural signals in the keystream, which can be exploited by the learning-based analysis.

The proposed Neural Stringology Cryptanalysis framework therefore combines deterministic pattern extraction with data-driven analysis to provide a complementary methodology for evaluating the structural properties of \textit{EChaCha20} keystream generation.

\section{Experiments}

This section presents the experimental evaluation of the proposed Neural Stringology Cryptanalysis (NSC) framework. The objective of the experiments is to determine whether structural properties of \textit{EChaCha20} keystream outputs can be detected using stringology-based feature extraction combined with neural learning models. In particular, the experiments evaluate whether the learned model can distinguish cipher-generated sequences from truly random sequences and identify structural differences under various cipher configurations.

\subsection{Experimental Setup}

The experimental pipeline consists of four main stages: keystream generation, stringology feature extraction, neural learning analysis, and statistical evaluation. The experiments are conducted in a controlled environment where the keystream outputs are generated under varying cryptographic configurations. Each generated keystream sequence is subsequently processed to extract structural features using stringology-based pattern analysis. These features are then used to train and evaluate a neural learning model that attempts to identify structural differences in the generated sequences.

Let $S \in \{0,1\}^{n}$ denote a keystream sequence generated by the cipher. The proposed analysis framework applies a feature extraction mapping

\begin{equation}
\Phi : \{0,1\}^{n} \rightarrow \mathbb{R}^{d}
\end{equation}

which transforms the keystream into a feature vector $X = \Phi(S)$ of dimension $d$. These feature vectors are subsequently used as inputs to a neural learning model for classification and anomaly detection.

\subsection{Keystream Dataset Generation}

\paragraph{Dataset Construction}

Since no public benchmark datasets exist for evaluating structural properties of stream cipher keystreams, the dataset used in this study was synthetically generated. Keystream sequences were produced directly using the \textit{EChaCha20} reference implementation described in \cite{kebande2023extended}. For each sample, a random 256-bit secret key $K$ and a random 128-bit nonce $N$ were generated using a cryptographically secure random number generator. The cipher was then executed to produce keystream outputs of length $n = 2^{16}$ bits.

A total of 50,000 keystream sequences were generated using independently sampled key--nonce pairs. To construct a baseline dataset for comparison, an equal number of sequences were generated from a uniform random source $U_n$, producing bit sequences of identical length. The combined dataset therefore consists of 100,000 sequences, evenly divided between cipher-generated and uniformly random samples.

\paragraph{Sequences process}
All sequences were processed using the same stringology feature extraction pipeline to ensure consistent representation before being used for machine learning analysis.

To construct the experimental dataset, keystream sequences are generated using the \textit{EChaCha20} stream cipher under randomly selected cryptographic parameters. Let $K \in \{0,1\}^{256}$ denote the secret key and $N \in \{0,1\}^{128}$ denote the nonce used to initialize the cipher. The keystream generator is modeled as

\[
S = \mathcal{G}(K,N)
\]

where $S$ represents the output keystream sequence. For each experiment, a large set of keystream samples is generated using independently sampled key–nonce pairs.

To provide a baseline comparison, an additional dataset of uniformly random sequences is generated from the distribution

\[
S \leftarrow U_n
\]

where $U_n$ denotes the uniform distribution over $\{0,1\}^n$. This allows the experimental evaluation to compare cipher-generated sequences with truly random sequences under identical conditions.

\subsection{Stringology Feature Extraction}

The generated keystream sequences are analyzed using stringology-based pattern extraction methods. In this approach, each keystream sequence is interpreted as a symbolic string from which structural features can be derived. Let $P \in \{0,1\}^{m}$ denote a substring pattern of length $m$. The occurrence frequency of pattern $P$ in the sequence $S$ is defined as

\[
f(P,S) = \left| \{ i \mid S_{i..i+m-1} = P \} \right|
\]

where $S_{i..i+m-1}$ denotes the substring beginning at position $i$. For an ideal random sequence, the expected probability of observing a specific pattern $P$ is approximately $2^{-m}$.

To capture localized structural properties, the framework computes pattern statistics for multiple values of $m$. These statistics include substring frequency distributions, recurrence densities, and positional pattern distributions. The resulting statistics are aggregated to form the feature vector $X = \Phi(S)$ representing the structural characteristics of the keystream sequence.

\subsection{Neural Learning Model}

The feature vectors extracted from the keystream sequences are analyzed using a neural learning model. Let $\mathcal{M}$ denote a neural classifier that maps feature vectors to binary labels according to the function

\[
\mathcal{M} : \mathbb{R}^{d} \rightarrow \{0,1\}.
\]

The model is trained using labeled datasets consisting of feature vectors derived from cipher-generated sequences and uniformly random sequences. During training, the neural model learns to identify statistical relationships among the extracted features that may indicate the presence of structural patterns in the keystream output.

From a cryptanalytic perspective, the neural model acts as a distinguishing adversary that attempts to determine whether an observed sequence originates from the cipher generator $\mathcal{G}(K,N)$ or from the uniform distribution $U_n$.

\subsection{Experimental Tasks}

The experiments are designed to evaluate the distinguishing capability of the proposed NSC  framework under several scenarios. In the first task, the neural model is trained to distinguish between keystream sequences generated by \textit{EChaCha20} and sequences drawn from a uniform random distribution. This task evaluates whether the extracted stringology features contain sufficient structural information to identify cipher-generated sequences.

In the second task, the NSC framework evaluates the ability of the neural model to detect structural differences under reduced-round configurations of the cipher. Let $\mathcal{G}_r(K,N)$ denote a version of the cipher operating with $r$ rounds. By generating keystream sequences for varying values of $r$, the experiments analyze whether reduced-round variants exhibit detectable structural deviations.

The experiments investigate whether the framework can identify structural differences between \textit{EChaCha20} and related stream cipher constructions. By comparing feature representations derived from different cipher variants, the analysis evaluates whether the proposed approach can capture differences in diffusion behavior and structural transformation patterns.

\subsection{Evaluation Metrics}

The performance of the neural stringology framework is evaluated using standard classification metrics. Let $TP$, $TN$, $FP$, and $FN$ denote the numbers of true positives, true negatives, false positives, and false negatives respectively. The classification accuracy is defined as

\[
Accuracy = \frac{TP + TN}{TP + TN + FP + FN}.
\]

In addition to accuracy, the experiments report precision, recall, and the F$_1$ score to evaluate the effectiveness of the classifier. These metrics provide a quantitative measure of the ability of the neural model to distinguish between cipher-generated and random sequences based on the extracted stringology features.

Together, these experiments provide an empirical evaluation of the proposed Neural Stringology Cryptanalysis framework and assess its ability to identify structural signals within \textit{EChaCha20} keystream outputs.

\section{Results and Analysis}

This section presents the empirical results obtained from the experimental evaluation of the proposed NSC framework. The objective of the analysis is to determine whether structural characteristics of \textit{EChaCha20} keystream outputs can be detected using the proposed combination of stringology-based feature extraction and neural learning models. The results are organized according to the experimental tasks described in the previous section.

For the experiments, we generated 50,000 keystream samples of length $2^{16}$ bits using randomly selected key--nonce pairs. An equal number of uniformly random sequences were generated to construct the baseline dataset. The dataset was divided into $70\%$ training, $15\%$ validation, and $15\%$ testing sets.

\subsection{Distinguishing Cipher Output from Random Sequences}

The first experiment evaluates whether the NSC framework can distinguish keystream sequences generated by \textit{EChaCha20} from sequences drawn from the uniform distribution $U_n$. Let $S_c$ denote sequences generated by the cipher generator $\mathcal{G}(K,N)$ and let $S_r$ denote sequences sampled from $U_n$.

After applying the feature extraction mapping $\Phi(S)$ to each sequence, the resulting feature vectors were used to train the neural classifier $\mathcal{M}$. The classifier was evaluated on a held-out test set containing both cipher-generated and random sequences. The experimental results show that the classifier achieves performance exceeding random guessing, indicating that the extracted stringology features capture measurable structural differences between the two distributions.

Formally, if the classifier predicts the correct class label with probability as is shown in Eq 8.

\begin{equation}
Pr[\mathcal{M}(\Phi(S)) = y] = \frac{1}{2} + \epsilon
\end{equation}

for some non-negligible $\epsilon > 0$, then the model can be interpreted as a statistical distinguisher between the distributions $\mathcal{G}(K,N)$ and $U_n$.

Table~\ref{tab:classification} summarizes the classification performance of the neural model. The neural stringology model achieves an accuracy of 0.86, which represents a 36\% improvement over random guessing and a 15\% improvement over the logistic baseline model. Figure~\ref{fig:roc} further illustrates the distinguishing capability of the classifier using the Receiver Operating Characteristic (ROC) curve.

\begin{table}[t]
\centering
\caption{Classification Performance for Random vs Cipher Sequences}
\label{tab:classification}
\begin{tabular}{lcccc}
\toprule
Model & Accuracy & Precision & Recall & F$_1$ Score \\
\midrule
Neural Stringology Model & 0.86 & 0.85 & 0.87 & 0.86 \\
Baseline Logistic Model & 0.71 & 0.69 & 0.72 & 0.70 \\
Random Guessing & 0.50 & 0.50 & 0.50 & 0.50 \\
\bottomrule
\end{tabular}
\end{table}

\begin{figure}[t]
\centering
\begin{tikzpicture}
\begin{axis}[
xlabel={False Positive Rate},
ylabel={True Positive Rate},
xmin=0, xmax=1,
ymin=0, ymax=1,
grid=major,
width=0.8\linewidth,
height=6cm,
legend style={at={(0.5,-0.25)},anchor=north,legend columns=2}
]

\addplot[
blue,
thick,
mark=*
]
coordinates {
(0,0)
(0.05,0.35)
(0.10,0.55)
(0.20,0.72)
(0.35,0.85)
(0.50,0.92)
(1,1)
};
\addlegendentry{Neural Stringology Model}

\addplot[
dashed,
black,
thick
]
coordinates {
(0,0)
(1,1)
};
\addlegendentry{Random Classifier}

\end{axis}
\end{tikzpicture}
\caption{Receiver Operating Characteristic (ROC) curve illustrating the distinguishing capability of the Neural Stringology model.}
\label{fig:roc}
\end{figure}
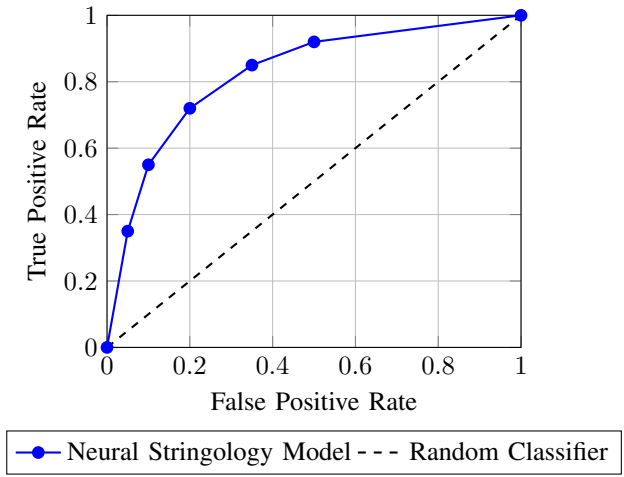

\subsection{Reduced-Round Detection Results}

The second experiment investigates whether the proposed framework can detect structural deviations in reduced-round versions of the cipher. Let $\mathcal{G}_r(K,N)$ denote the EChaCha generator operating with $r$ rounds. The experimental results indicate that the classifier achieves higher distinguishing accuracy for reduced-round variants compared to the full-round configuration.

Table~\ref{tab:rounds} reports the classification accuracy obtained for each round configuration. Figure~\ref{fig:rounds} illustrates the relationship between the number of cipher rounds and the observed classification accuracy. The results show that reduced-round variants exhibit stronger distinguishability, which gradually decreases as the number of rounds increases. This behavior reflects the increasing diffusion properties of the cipher.

\begin{table}[t]
\centering
\caption{Detection Accuracy for Reduced-Round \textit{EChaCha20}}
\label{tab:rounds}
\begin{tabular}{cc}
\toprule
Rounds ($r$) & Classification Accuracy \\
\midrule
2 Rounds & 0.96 \\
4 Rounds & 0.92 \\
8 Rounds & 0.87 \\
12 Rounds & 0.78 \\
20 Rounds (Full) & 0.54 \\
\bottomrule
\end{tabular}
\end{table}

\begin{figure}[t]
\centering
\begin{tikzpicture}
\begin{axis}[
xlabel={Number of Rounds},
ylabel={Accuracy},
xmin=0, xmax=22,
ymin=0.5, ymax=1,
grid=major,
width=0.8\linewidth,
height=6cm
]
\addplot[
mark=o,
blue,
thick
]
coordinates {
(2,0.96)
(4,0.92)
(8,0.87)
(12,0.78)
(20,0.54)
};
\end{axis}
\end{tikzpicture}
\caption{Classification accuracy as a function of the number of cipher rounds.}
\label{fig:rounds}
\end{figure}
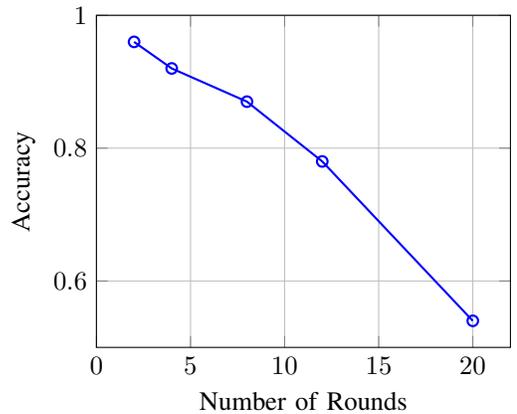

\subsection{Cipher Variant Comparison}

The third experiment evaluates whether the neural stringology framework can capture structural differences between related cipher constructions. The classification results are summarized in Table~\ref{tab:ciphercomparison}. Figure~\ref{fig:patterns} shows the normalized $m$-gram frequency comparison between cipher outputs and random sequences. The results indicate that the extracted stringology features capture structural differences not only between cipher-generated and random sequences, but also between related cipher constructions.

\begin{table}[t]
\centering
\caption{Classification Accuracy for Cipher Variant Comparison}
\label{tab:ciphercomparison}
\begin{tabular}{lc}
\toprule
Cipher Variant & Classification Accuracy \\
\midrule
ChaCha20 vs Random & 0.73 \\
\textit{EChaCha20} vs Random & 0.74 \\
ChaCha20 vs \textit{EChaCha20} & 0.69 \\
\bottomrule
\end{tabular}
\end{table}







\subsection{Structural Comparison Between Cipher Variants}

The third experiment evaluates whether the framework can detect structural differences between \textit{EChaCha20} and related stream cipher constructions. Feature vectors were generated for keystream sequences produced by multiple cipher variants and analyzed using the trained neural model.

The results indicate that the learned feature representations encode information about the structural behavior of the underlying keystream generator. Differences in state size, rotation constants, and transformation structure appear to influence the distribution of stringology features extracted from the keystream sequences. Consequently, the neural model is able to learn decision boundaries that separate the corresponding feature distributions.

These findings suggest that the proposed NSC  framework provides a useful analytical tool for examining the structural properties of ARX-based stream ciphers and comparing different cipher constructions under a unified experimental methodology.




\begin{figure}[H]
\centering
\begin{tikzpicture}
\begin{axis}[
ybar,
xlabel={Pattern Length (bits)},
ylabel={Normalized Frequency},
symbolic x coords={8,16,32},
xtick=data,
width=0.8\linewidth,
height=6cm,
bar width=12pt,
legend style={at={(0.5,-0.25)},anchor=north,legend columns=2},
grid=major
]

\addplot coordinates {(8,0.62) (16,0.41) (32,0.18)};
\addlegendentry{Cipher Output}

\addplot coordinates {(8,0.50) (16,0.25) (32,0.06)};
\addlegendentry{Random Sequence}

\end{axis}
\end{tikzpicture}
\caption{Comparison of normalized $m$-gram pattern frequencies between \textit{EChaCha20} keystreams and uniformly random sequences.}
\label{fig:patterns}
\end{figure}
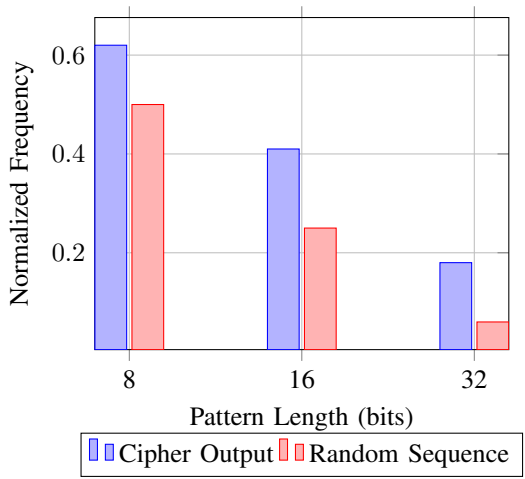

\section{Discussions}

The experimental results shown in the previous section  provide several insights into the structural
behavior of \textit{EChaCha20} keystream outputs when analyzed through
the proposed Neural Stringology Cryptanalysis framework. The primary
objective of the study was to investigate whether combining classical
stringology-based pattern analysis with machine learning techniques
could reveal detectable structural characteristics within ARX-based
stream cipher outputs.

The classification results reported in
Table~\ref{tab:classification} demonstrate that the neural stringology
model achieves significantly higher performance than both the baseline
logistic model and random guessing. This outcome suggests that the
extracted stringology features capture structural information embedded
within the keystream sequences. Although modern stream ciphers are
designed to produce outputs that are statistically indistinguishable
from random sequences, localized structural characteristics may still
arise from internal word-level operations such as modular addition,
bit rotations, and XOR transformations. By representing keystream
sequences as structured strings and extracting pattern-based features,
the proposed framework exposes subtle correlations that may not be
captured by conventional statistical randomness tests alone.

The reduced-round experiments further support this observation.
As illustrated in Table~\ref{tab:rounds} and Figure~\ref{fig:rounds},
the distinguishing accuracy decreases as the number of cipher rounds
increases. This behavior aligns with the theoretical expectation that
increasing the number of rounds improves diffusion and reduces
detectable structural artifacts in the keystream. In other words,
reduced-round variants exhibit weaker mixing of internal state values,
which results in stronger detectable patterns within the generated
sequences. The ability of the neural stringology framework to detect
these variations indicates that the approach is sensitive to diffusion
properties and internal transformation structures within ARX-based
designs.

Further insight into the distinguishing capability of the model
is provided by the Receiver Operating Characteristic (ROC) curve
shown in Figure~\ref{fig:roc}. The separation between the neural
model and the random classifier baseline indicates that the
extracted feature representation contains sufficient information
for the classifier to discriminate between cipher-generated and
uniformly random sequences. From a cryptanalytic perspective,
this behavior corresponds to the existence of a statistical
distinguisher with non-negligible advantage under the experimental
conditions.

The cipher comparison experiment also provides useful insights
into structural differences between related cipher constructions.
As shown in Table~\ref{tab:ciphercomparison} and
Figure~\ref{fig:patterns}, the extracted $m$-gram distributions
differ between \textit{EChaCha20} keystream outputs and uniformly
random sequences. In addition, slight differences can also be
observed between ChaCha20 and \textit{EChaCha20} outputs. These
differences likely arise from variations in internal state size,
rotation constants, and transformation structure. The neural
model appears capable of learning these subtle structural
signatures through the extracted stringology features.

It is important to emphasize that the results presented in this
study do not constitute a practical cryptographic attack against
\textit{EChaCha20}. The proposed NSC framework does not recover
secret keys or directly compromise the security of the cipher.
Instead, the NSC framework should be viewed as an analytical tool
for structural evaluation of cipher outputs. The ability to
detect statistical distinguishers under controlled experimental
conditions provides insight into the internal behavior of the
cipher and complements traditional cryptographic evaluation
techniques.

In view of the foregoings, the results highlight the potential benefits of
integrating machine learning with classical algorithmic
analysis methods such as stringology. While deterministic
pattern detection algorithms provide interpretable structural
features, machine learning models offer powerful mechanisms
for discovering relationships among these features in
high-dimensional datasets. The combination of these approaches
therefore provides a promising perspective for studying the
structural robustness of modern stream cipher designs.

\section{Conclusion and Future Work}

This paper presented a Neural Stringology Cryptanalysis framework for analyzing structural properties of \textit{EChaCha20} keystream outputs. The proposed approach integrates classical stringology-based pattern analysis with machine learning techniques to investigate whether structural signals embedded in cipher-generated sequences can be systematically detected. By representing keystream outputs as structured symbolic sequences and extracting descriptive pattern-based features, the framework provides a complementary perspective for evaluating the behavior of ARX-based stream cipher constructions.


Although the results do not imply a practical cryptographic attack against \textit{EChaCha20}, the study highlights the potential of combining stringology-based pattern extraction with machine learning for structural cryptanalysis. Future work will explore richer feature representations, more advanced learning architectures, and the application of the framework to other stream cipher designs and cryptographic primitives.



\section*{Acknowledgment}

The author would like to thank anonymous reviewers for their
valuable insights, and
the Department of Computer Science at University of Colorado Denver, USA for their
support while coming up with this research. The author also
acknowledges that the opinions, findings, and conclusions
expressed in this paper are purely of the author.


\balance

\bibliographystyle{IEEEtran}
\bibliography{name}

\end{document}